\documentclass[a4paper,onecolumn,11pt,accepted=2025-11-04]{quantumarticle}
\pdfoutput=1
\usepackage[numbers]{natbib}
\usepackage[utf8]{inputenc}
\usepackage{amsmath, amssymb, amsfonts,amsthm,amsopn,amscd}
\usepackage{dsfont}
\usepackage{graphicx}
\usepackage{booktabs}
\usepackage{longtable}
\usepackage{titletoc}
\usepackage[section]{placeins}
\usepackage[shortlabels]{enumitem}

\usepackage[centertableaux]{ytableau}
\usepackage{mathtools}

\usepackage{tikz-cd}
\usepackage{upgreek}
\usepackage{xcolor}
\usepackage[breaklinks]{hyperref}
\hypersetup{
    colorlinks,
    linkcolor={red!40!black},
    citecolor={blue!50!black},
    urlcolor={blue!20!black}
}

\DeclareMathOperator{\tr}{tr}

\DeclareMathOperator{\height}{height}

\DeclareMathOperator{\spect}{spect}
\DeclareMathOperator{\id}{id}
\DeclareMathOperator{\Herm}{Herm}
\DeclareMathOperator{\Pos}{Pos}

\newcommand{\ket}[1]{\mathinner{|#1\rangle}}

\newcommand{\ot}[0]{\otimes}
\newcommand{\one}[0]{\mathds{1}}
\renewcommand{\a}{\alpha}

\renewcommand{\SS}{\mathcal{S}}
\renewcommand{\AA}{\mathcal{A}}

\newcommand{\N}{\mathds{N}}

\newcommand{\C}{\mathds{C}}

\newcommand{\HH}{\mathcal{H}}

\newcommand{\XX}{\mathcal{X}}
\newcommand{\YY}{\mathcal{Y}}
\newcommand{\UU}{\mathcal{U}}
\newcommand{\PP}{\mathcal{P}}
\newcommand{\DD}{\mathcal{D}}

\newcommand{\frU}{\mathfrak{U}}

\newcommand{\cdn}[0]{(\C^d)^{\otimes n}}
\newcommand{\cdk}[0]{(\C^d)^{\otimes k}}
\newcommand{\cdkn}[0]{((\C^d)^{\ot k})^{\ot n}}

\newcommand{\csk}[0]{\C S_k}
\newcommand{\Aell}   [0]{\substack{ A \in \AA \\  \sigma \in S_k }}

\newtheorem{theorem}{Theorem}
\newtheorem*{theorem*}{Theorem}
\newtheorem{proposition}[theorem]{Proposition}

\newtheorem{corollary}[theorem]{Corollary}

\newtheorem{example}[theorem]{Example}
\newtheorem{remark}[theorem]{Remark}

\newtheorem*{problem*}{Problem}
\newtheorem*{question*}{Question}

\newtheorem*{result*}{Result}

\newtheorem{thmA}{Theorem}

\newcommand{\nn}{\nonumber}
\newcommand{\overbar}[1]{\mkern 1.5mu\overline{\mkern-1.5mu#1\mkern-1.5mu}\mkern 1.5mu}

\begin{document}

\title{Refuting spectral compatibility of quantum marginals}

\author{Felix Huber}
\affiliation{Division of Quantum Computing, Faculty of Mathematics, Physics and Informatics, University of Gdańsk, Wita Stwosza 57, 80-308 Gdańsk, Poland}
\orcid{0000-0002-3856-4018}
\email{felix.huber@ug.edu.pl}

\author{Nikolai Wyderka}
\affiliation{Institut für Theoretische Physik III, Heinrich-Heine-Universität Düsseldorf,Universitätsstraße 1, D-40225 Düsseldorf, Germany}
\orcid{0000-0003-3002-9878}
\email{wyderka@hhu.de}

\maketitle

\begin{abstract}
The spectral variant of the quantum marginal problem asks:
Given prescribed spectra for a set of overlapping quantum marginals,
does there exist a compatible joint state? 
The main idea of this work is a symmetry-reduced semidefinite programming hierarchy that detects when no such joint state exists.
The hierarchy is complete, in the sense that it detects every
incompatible set of spectra.
The refutations it provides are
dimension-free, certifying incompatibility in all local dimensions.
The hierarchy also applies to the
sums of Hermitian matrices problem, 
the compatibility of local unitary invariants,
for certifying vanishing Kronecker coefficients,
and to optimize over equivariant state polynomials.
\end{abstract}

\setcounter{tocdepth}{1}

\section{Introduction}
The compatibility of quantum marginals (also known as reduced density matrices) 
is central to quantum phenomena such as entanglement and non-locality.
It also plays a key role in quantum algorithms like quantum error correction and adiabatic quantum computation.
At the heart of this {\em quantum marginal problem} lies a 
constraint satisfaction problem with prohibitive computational complexity:
it is QMA-complete, with QMA being the quantum analog of NP~\cite{Liu2006}.
This renders 
molecular-structure and ground state calculations in chemistry\footnote{
For fermionic systems, the quantum marginal problem is also known as the $N$-representability problem,
the full set of conditions for two-electron reduced states were given by Mazziotti 
in~\cite{PhysRevLett.108.263002}.}
and physics challenging.
Consequently, 
a large literature focuses on conditions for quantum marginals to be compatible~\cite{
Schilling2014, 
https://doi.org/10.3929/ethz-a-010250985,
Huber_Thesis,
Klassen_phd-thesis}.

A more fundamental problem is to decide on the compatibility of the spectra instead of the reduced density matrices.
The simplest formulation of this {\em spectral} variant of the quantum marginal problem
is perhaps the following:
Given a collection of prescribed eigenvalues $\mu_{AB}$ and $\mu_{BC}$, associated with subsystems $AB$ and $BC$, respectively,
does there exist a joint state $\varrho_{ABC}$ such that its reduced density matrices 
$\varrho_{AB} = \tr_C(\varrho_{ABC})$ and $\varrho_{BC} = \tr_A(\varrho_{ABC})$ have 
spectra $\spect(\varrho_{AB}) = \mu_{AB}$ and $\spect(\varrho_{BC}) = \mu_{BC}$? 
If such a joint state exists, the spectra are said to be {compatible};
they are {incompatible} otherwise.
This spectral formulation also maintains an intimate connection to 
fundamental questions in representation and matrix theory~\cite{Klyachko1998}.

The work by Klyachko~\cite{1742-6596-36-1-014} allows one
to establish {\it compatibility} of prescribed spectra 
in the bipartite case through representation theoretic methods:
compatible spectra correspond to families of non-vanishing dilated Kronecker coefficients.
For tripartite systems,
Christandl, {\c{S}}ahino{\u{g}}lu, and Walter~\cite{Christandl2018} showed that
spectra are compatible if and only if the recoupling coefficients of the symmetric group
$S_k$ decay at most polynomially in $k$.
Kronecker coefficients can be computed with algorithms from
algebraic combinatorics and geometric complexity theory~\cite{BALDONI2018113, 8555166},
giving rise to a hierarchy of one-sided criteria for compatibility for the non-overlapping case.
Complete lists of inequalities for non-overlapping spectra of bi- and tripartite systems are given in Refs.~\cite{1742-6596-36-1-014, VergneWalter2017}.
It is harder, however, for these methods to determine the {\it incompatibility}
of marginal spectra, in particular when they overlap.

The aim of this manuscript is to provide such a complementary method:
a semidefinite programming hierarchy for certifying spectral incompatibility,
where the marginals are allowed to overlap (Section~\ref{sec:SDPref}).
It is complete, in the sense that it detects every set of incompatible spectra at some level of the hierarchy.
Our formulation in terms of a symmetric extension hierarchy is furthermore symmetry-reduced,
drastically reducing the size of the optimization problem (Section~\ref{sec:symmetry}).
This approach can produce spectral incompatibility certificates for both finite fixed local dimensions 
and arbitrary local dimensions (Section~\ref{sec:dimfree}).
A modern desktop computer can access up to the fourth level of the hierarchy in the case of four-partite states,
and the fifth level in the case of tripartite states (Section~\ref{sec:numerics}).

\section{Contribution}
Let $\varrho \in L(\cdn)$ be an $n$-partite quantum state of local dimension $d$
and $\AA$ a collection of subsystems of $\{1,\dots, n\}$.
Given a subsystem $A \in \AA $, 
denote by $\varrho_A = \tr_{A^c}(\varrho) $ the reduced density matrix on $A$ 
and by $\mu_A$ the eigenvalues of $\varrho_A$, i.e., the spectrum on $A$.
We want to answer the following:

\begin{problem*}\label{prob:sqmp}
Let $\AA$ be a collection of subsets of $\{1,\dots,n\}$.
Given prescribed spectra $\{ \mu_A \, |\, A \in \AA \}$,
does there exist a joint state $\varrho$ for which
the spectrum of $\varrho_A = \tr_{A^c}(\varrho)$ equals $\mu_A$
for all $A \in \AA$? 
\end{problem*}

We provide the following symmetry-reduced semidefinite programming hierarchy for determining spectral incompatibility
for overlapping marginals. 
\begin{thmA}
Let $\AA$ be a collection of subsets of $\{1,\dots,n\}$
with associated marginal spectra $\{ \mu_A \, |\, A \in \AA \}$.
The spectra are compatible with a joint quantum state on $(\C^d)^{\ot n}$ 
if and only if every level in the hierarchy~\eqref{eq:hierarchy_dual_symred} is feasible.
If a level of the hierarchy 
returns a negative value, then the spectra are incompatible.
\end{thmA}
For the proof, see Theorem~\ref{thm:hierarchy_dual_symred} and Theorem~\ref{thm:complete}.
The symmetry-reduction allows to work with up to the fourth level of the hierarchy
for four-partite systems
and the fifth level for three-partite systems on a modern desktop computer
(see Figure~\ref{fig:threequbit_incompat_m4} and Table~\ref{fig:flat_spectra}).

\begin{thmA}
When the level of the hierarchy is less than or equal to the local dimension ($k\leq d$),
the incompatibility witnesses produced by the hierarchy~\eqref{eq:hierarchy_dual_symred} 
are dimension-free
and the spectra are incompatible in all local dimensions.
\end{thmA}

For the proof, see Theorem~\ref{thm:dimfree}.
As a consequence, the SDP refutations stabilize 
when the level of the hierarchy $k$ equals the local dimension,
certifying incompatibility for all local dimensions.

\subsection{Further applications}
We list further areas to which  our complete hierarchy also applies.
These are either reformulations or special cases of the spectral marginal problem such as
$(1)$ and $(2)$; or problems that are slightly more general such as $(3)$ and $(4)$.

\begin{enumerate}
\item {\bf Kronecker and recoupling coefficients.}
Klyachko has shown that $\spect(\varrho_A)$, $\spect(\varrho_B)$, and $\spect(\varrho_{AB})$ are
compatible if and only if dilations of associated Young tableaux $\lambda, \mu, \nu$ 
allow for a non-vanishing Kronecker coefficient, $g(m\lambda, m\mu, m\nu)$ $\neq 0$ 
for some $m > 0$~\cite{quant-ph/0409113v1}.\footnote{
See also the work by Christandl and Mitchison~\cite{Christandl2006} 
that showed one direction of this statement.}
A related statement holds for marginals of tripartite systems:
the recoupling coefficients of the symmetric group $S_k$ decay polynomially in $k$ if the spectra are compatible and exponentially otherwise~\cite{Christandl2018}.
Deciding the positivity of Kronecker coefficients is an NP-hard task~\cite{Ikenmeyer2017}.
The algorithm by Baldoni, Vergne, and Walter allows to compute dilated Kronecker coefficients~\cite{BALDONI2018113},
giving rise to a hierarchy of one-sided compatibility criteria for the non-overlapping case.
Our hierarchy provides a complementary method:
showing that a set of spectra is incompatible also proves that
the Kronecker coefficient, $g(m\lambda, m\mu, m\nu) = 0$ for all $m \in \N$.
Thus, our hierarchy is also complete for this problem.

\smallskip
\item {\bf Sums of Hermitian Matrices.}
The Sums of Hermitian Matrices problem (solved by Klyachko~\cite{Klyachko1998} 
and related to honeycombs by Knutson and Tao~\cite{KnutsonTao2001}) asks:
given Hermitian matrices $A$ and $B$ with 
spectra $\spect(A)$ and $\spect(B)$, what are the constraints on the spectrum of $A+B$?
It can be shown that this problem is equivalent to an instance of the one-body spectral
marginal problem~\cite{Klyachko1998, quant-ph/0409113v1}.
For the related question for known spectra of sums of multiple matrices, e.g., $A+B$ and $A+C$, a similar relation in terms of overlapping marginals holds
\cite[Lemma 9.13.]{https://doi.org/10.3929/ethz-a-010250985},
~\cite[Theorem 12]{Christandl2018}.
Our hierarchy can certify 
that such spectra cannot be realized
by sums of hermitian matrices with known spectra.

\smallskip
\item {\bf Local unitary invariants and quantum codes.}
Local unitary invariants of $n$-partite quantum states $\varrho$
correspond to the expectation values of elements of $\C S_k^n$ on
$\varrho^{\otimes k}$.
Modifying our SDP to include these more general objective functions and constraints,
one can certify the incompatibility of a set of local unitary invariants with a joint quantum state.
This strengthens the linear programming bounds on quantum codes:
the existence of a quantum code of given block length, size, and distance can be formulated
in terms of a compatible set of local unitary invariants of degree two,
the quantum weight enumerators~\cite{796376}.
Our hierarchy can certify that no compatible weight enumerators exist, thus ruling out the existence of a corresponding quantum code.

\smallskip
\item {\bf Equivariant state polynomials.}
Our hierarchy allows to numerically find inequalities in the Löwner order for equivariant state polynomials.\footnote{
Note that this setting is distinct from the one found in Ref.~\cite{klep2023statepolynomialspositivityoptimization},
which refers to non-commutative polynomials evaluated on states. Also, $\tr(\,\cdot\,)$ refers to the non-normalized matrix trace,
in contrast to~\cite{klep2021optimization}.}
These constitute a type of polynomials whose variables are states, and whose positivity is invariant under local unitary transformations.
For example,
$(\varrho^{T_1}\sigma^{T_1})^{T_1}$
is an equivariant state polynomial in two bipartite states $\varrho$ and $\sigma$ where $(\,\cdot\,)^{T_1}$ is the partial transpose.
Our hierarchy can minimize such expressions [e.g., through a converging sequence on lower bounds on $\a = \min_{\varrho,\sigma,\mu} \tr\big( (\varrho^{T_1} \sigma^{T_1})^{T_1}\mu\big)$], so that $(\varrho^{T_1} \sigma^{T_1})^{T_1} - \a\one$ is positive semidefinite for all states $\varrho,\sigma$. This allows us to systematically find new equivariant state polynomial inequalities.
\end{enumerate}

\section{Notation}

\subsection{Quantum systems} 
Denote by $L(\HH)$ the space of linear maps acting on a Hilbert space $\HH$.
Quantum states on $n$ systems with $d$ levels each are represented by 
positive operators of trace one acting on $\cdn$, i.e., satisfying 
 $\varrho \geq 0$, $\tr(\varrho) = 1$.
The marginal or reduced state of a $n$-partite state $\varrho$ on subsystems $A \subseteq \{1,\ldots,n\}$
is denoted by
$\varrho_A = \tr_{A^c}(\varrho)$, where $A^c$ is the complement of $A$ in $\{1,\dots, n\}$.
In what follows, $\AA$ is a collection of subsets $A \subseteq \{1,\dots,n\}$.
We make use of the coordinate-free definition of the partial trace, which states that the partial trace $\tr_2$ over the second of two systems $\HH_1 \otimes \HH_2$ is the unique linear operator
satisfying
\begin{equation}\label{eq:ptrace}
\tr\big((M\ot \one) N\big) = \tr\big(M \tr_2(N)\big)
\end{equation} 
for all $M\in L(\HH_1)$ and $N \in L(\HH_1 \ot \HH_2)$.
Finally, we denote the set of unitary $d \times d$ matrices by $\UU(d)$.

\subsection{Symmetric group}

Our work makes use of $k$ copies of $n$-particle states, with the symmetric group acting on both copies and their subsystems.
The symmetric group permuting $k$ elements is $S_k$.
The group ring $\C S_k$ is formed by formal sums
$\csk = \big\{\sum_{\sigma \in S_k} a_\sigma \sigma \,:\, a_\sigma \in \C \big\}$.
The linear extension of the multiplication of $S_k$ defines the multiplication on $\csk$.
An element $a$ of $\csk$ with $a = \sum_{\sigma \in S_k} a_\sigma \sigma$, has an assoicated adjoint
$a^* = \sum_{\sigma \in S_k} \overbar{a}_\sigma \sigma^{-1}$;
it is Hermitian if $a = a^*$.

Denote by $S_k^n = S_k \times \dots \times S_k$ the $n$-fold Cartesian product of $S_k$.
Let $\pi \in S_k$ act on $\sigma = (\sigma_1, \dots, \sigma_n) \in S_k^n$ via
\begin{equation}\label{eq:outer_rep}
 \pi \sigma \pi^{-1} := (\pi \sigma_1 \pi^{-1}, \dots, \pi \sigma_n \pi^{-1}),
\end{equation} 
and by linear extension also on $\C S_k^n$.
Finally, denote by
$ (\C S_k^n)^{S_k}$
the subspace of $\C S_k^n$ invariant under the diagonal action of $S_k$,
\begin{equation}
(\C S_k^n)^{S_k} = 
\big\{ 
a \in \C S_k^n 
\,:\,
a = \pi a \pi^{-1}\,\forall 
        \pi \in S_k \,
\big\}.
\end{equation}

\subsection{Representations}

Let $\sigma \in S_k$ act on $\cdk$ by its representation $\eta_d(\sigma)$, which permutes the tensor factors,
\begin{equation}\label{eq:eta_d}
\eta_d(\sigma) 
\ket{v_1} 
\ot \cdots \ot 
\ket{v_k} 
= 
\ket{v_{\sigma^{-1}(1)}} 
\ot \cdots \ot 
\ket{v_{\sigma^{-1}(k)}}\,.
\end{equation}
Now consider $\sigma = (\sigma_1, \dots, \sigma_n) \in S_k^n$.
It acts on $((\C^d)^{\ot k})^{\ot n}$ 
via
\begin{equation}
 \eta^d(\sigma) := \eta_d(\sigma_1) \ot \dots \ot \eta_d(\sigma_n)\,.
\end{equation}
with $\eta_d(\sigma_i)$ acting on the collection of the $k$ copies of the $i$'th tensor factor.\footnote{
This is the same setting as found in Ref.~\cite{817508}.}
If the local dimension $d$ is clear from the context, we will use $\eta = \eta^d$.

Another action we need is that of permuting the $k$ copies of $n$-partite states. 
For $\pi \in S_k$, the representation $\tau$ 
acts diagonal conjugate on $((\C^d)^{\ot k})^{\ot n}$
\begin{equation}
 \tau(\pi) \eta(\sigma) \tau(\pi^{-1}):= \eta_d(\pi \sigma_1 \pi^{-1}) \ot \dots \ot \eta_d(\pi \sigma_n \pi^{-1})
 = \eta(\pi \sigma \pi^{-1})\,,
\end{equation}
making it compatible with Eq.~\eqref{eq:outer_rep}.

Finally, a representation $R$ is called orthogonal if $R(g^{-1}) = R(g)^T$.
For the symmetric group, Young's orthogonal representation is orthogonal~\cite{JamesKerber1984}.
In the software {\em SageMath}, it can be obtained with the command \texttt{SymmetricGroupRepresentation} ~\cite{sagemath}.

\section{Spectra are polynomial in $\varrho$} 

\subsection{Spectrum from $\varrho^{\ot k}$}
We first show how the spectral quantum marginal problem can be formulated as a 
constraint that is polynomial in $\varrho$.
This is done with a generalization of the so-called swap trick. 

Consider a single quantum system $\varrho \in L(\C^d)$.
It is clear that $\tr(\varrho^\ell) = \sum_{i} \mu_i^\ell$,
where $\mu_i$ are the eigenvalues of $\varrho$.
A complex $d \times d$ matrix has $d$ eigenvalues, such that 
the set $\{\tr(\varrho^\ell)\,:\, \ell = 1, \dots, d\}$ determines
the spectrum of $\varrho$.

Recall that $\sigma \in S_k$ acts on $\cdk$ via its representation $\eta_d(\sigma)$ that permutes the tensor factors,
\begin{equation}
\eta_d(\sigma) 
\ket{v_1}   
\ot \cdots \ot 
\ket{v_k} 
= 
\ket{v_{\sigma^{-1}(1)}} 
\ot \cdots \ot 
\ket{v_{\sigma^{-1}(k)}}\,.
\end{equation} 
For a cycle $(\alpha_1 \dots \alpha_\ell) \in S_k$ of length $\ell \leq k$ and a Hermitian matrix $B \in L(\C^d)$, 
it is known that
  $\tr( \eta_d((\alpha_1 \dots \alpha_\ell) ) B^{\otimes k}) = \tr(B^{\ell}) \tr(B)^{k-\ell}$~\cite{Kostant1958}.
For a density matrix, this simplifies further to 
\begin{equation}\label{eq:invariants_from_perm0}
 \tr( \eta_d((\alpha_1 \dots \alpha_\ell) ) \varrho^{\otimes k})
 = \tr(\varrho^{\ell})\,.
\end{equation}
Consequently, under a global trace, the permutation operators acting 
on copies of a state~$\varrho$ can recover its spectrum.

\subsection{Permuting subsystems of copies}
A similar strategy works with multipartite states. 
Then, we additionally need to consider the action of permutations on subsystems.

Recall that the element $\sigma \in S_k^n$ acts on $((\C^d)^{\otimes k})^{\otimes n}$
via
\begin{equation}\label{eq:perm_on_subs}
 \eta(\sigma) := \eta_d(\sigma_1) \ot \dots \ot \eta_d(\sigma_n)\,,
\end{equation}
with $\eta_d(\sigma_i)$ acting on the $k$ copies of the $i$'th tensor factor.
Now, for a subset $A \subseteq \{1, \dots, n\}$, define $\sigma^A = (\sigma^A_1, \dots, \sigma^A_n) \in S_k^n$ through
\begin{equation}
 \sigma^A_i = 
 \begin{cases}
  \sigma \quad \,\,\mathrm{if }~i \in A \\
  \id    \quad \mathrm{if }~i \not \in A\,.
 \end{cases}
\end{equation} 
By Eq.~\eqref{eq:perm_on_subs}, the operator $\eta(\sigma^A)$ acts 
on the collection of subsystems contained in $A$ with~$\sigma$, 
while it acts with the identity matrix on the remaining subsystems.

With some abuse of notation, 
$\eta(\sigma^A)$ can be thought of acting on $((\C^d)^{\otimes n})^{\otimes k}$ 
as well as on any tensor space containing the subsystems~$A$. 
For $\ell \leq k$,
Eq.~\eqref{eq:invariants_from_perm0} generalizes to
\begin{align}\label{eq:invariants_from_perm}
 \tr\big( \eta((\alpha_1\dots \alpha_\ell)^A) \varrho^{\otimes k} \big)
 &= \tr\big(\eta((\alpha_1\dots \alpha_\ell)^A) \tr_{A^c} (\varrho^{\otimes k}) \big) \nn\\
 &= \tr\big(\eta((\alpha_1\dots \alpha_\ell)^A) \varrho_A^{\otimes \ell} \big) 
 = \tr\big(\varrho_A^{\ell} \big)\,.
\end{align}
where we have used the coordinate-free definition of the partial trace in Eq.~\eqref{eq:ptrace}.

\smallskip
Let a prescribed spectrum $\mu_A$ on subsystem $A$ be given.
Define 
\begin{equation}\label{eq:q_to_lambda}
q_{A, \ell} = \sum_{\mu_i \in \mu_A} \mu_i^\ell\,.
\end{equation}
If a $\varrho$ realizing $\mu_A$ exists, then 
for any $\sigma = (\a_1\dots \a_\ell) \in S_k$,
\begin{equation}\label{eq:qAell}
 q_{A, \ell} = \tr(\eta(\sigma^A) \varrho^{\ot k}) = \tr( \varrho_A^\ell)\,.
\end{equation} 
Allowing for more general $\sigma \in S_k^n$, one can generalize this to incorporate any local unitary invariant polynomial function of reductions of the state via
\begin{equation}
 q_{A, \sigma} = \tr(\eta(\sigma^A) \varrho^{\ot k})\,.
\end{equation}

\subsection{Compatibility conditions}

Denote by $\HH = \cdn$ the space of a $n$-qudit system.
Our discussion makes the following immediate.

\begin{proposition}\label{prop:qmp1}
 Let $\AA$ be a collection of subsystems of $\{1,\ldots,n\}$ and $\mu = \{ \mu_A \,|\, A \in \AA\}$ be prescribed  spectra of reductions.
 Let $m$ be the size of the largest spectrum
 and 
 $q_{A, \ell}$ be given in terms of $\mu_A$ by Eq.~\eqref{eq:q_to_lambda}.
 Then $\mu$ is compatible with a joint state,
 if and only if 
 there exists a state $\varrho \in L(\HH)$,
 such that for all 
 $\ell$-cycles $\sigma = (\a_1, \dots, \a_\ell)$
 with $\ell = 1, \dots, m$,
 and $A \in \AA$,
 \begin{align}\label{eq:power-sums}
  \tr(\eta(\sigma^A) \varrho^{\otimes m} ) &= q_{A, \ell}\,.
 \end{align}
\end{proposition}
\begin{proof}
 If a compatible $\varrho$ exists, then $\tr(\eta(\sigma)^A \varrho^{\otimes m})$ evaluates through Eq.~\eqref{eq:qAell} to $q_{A,\ell}$.
 Conversely, if there exists a $\varrho$ satisfying Eq.~\eqref{eq:power-sums} for all $\ell$-cycles and $A\in \AA$,
 then its spectrum on $A$ is completely determined and equal to $\mu_A$ for all $A\in \AA$.
\end{proof}

\subsection{Symmetric extension relaxation}

We relax the tensor product $\varrho^{\otimes m}$ in Proposition~\ref{prop:qmp1} to a symmetric state:
 \begin{proposition}\label{prop:qmp2}
 Let $\AA$ be a collection of subsystems of $\{1,\ldots,n\}$ and $\mu = \{ \mu_A \,|\, A \in \AA\}$ 
 be prescribed spectra of reductions. 
 If the spectra $\mu$
 are compatible with a joint state,
 then for every $k\in \N$ 
 there exists a state $\varrho_k \in L(\HH^{\otimes k})$
such that for all 
 $\ell$-cycles $\sigma = (\a_1, \dots, \a_\ell)$
 with $\ell = 1, \dots, k$,
 and $A \in \AA$,
\begin{equation}\label{eq:nec_crit}
  \begin{aligned} 
 \tr\big( \eta(\sigma^A) \varrho_k \big)
                &= q_{A, \ell}\,.
   \end{aligned}
\end{equation}
\end{proposition}
It is clear that the constraints in Proposition~\ref{prop:qmp2} are weaker than those in Proposition~\ref{prop:qmp1}.

\begin{remark}\label{rmk:PPT}
  One could add the constraint of a positive partial transpose $\varrho_k^{T_R}  \geq0\,, \forall R \subseteq \{1,\dots, n\}$
  to Proposition~\ref{prop:qmp2}. 
  However, this approach is not directly suitable to the
  symmetry reduction method employed in this manuscript.
\end{remark}

\subsection{Invariance}\label{sec:invariance}
One can see that if $\varrho_k$ satisfies Eq.~\eqref{eq:nec_crit},
then so do the states in the set 
\begin{equation}\label{eq:U-inv}
    \big\{(U_1 \ot \dots \ot U_{n})^{\ot k} \,\varrho_k \, ((U_1 \ot \dots \ot U_{n})^\dag)^{\ot k}\,:\,
    U_1, \dots, U_n \in \UU(d) \big\}\,.
\end{equation}
This can be understood from the fact that the eigenvalues of a matrix are unitary invariants.
As a second invariance, also the states in
\begin{equation}\label{eq:S-inv}
     \big \{ \tau(\pi) \varrho_k \tau(\pi)^{-1}\,:\, \pi \in S_k \big\}\,,
\end{equation}
where $\tau(\pi)$ acts diagonally on $(\cdk)^{\ot n}$, satisfy Eq.~\eqref{eq:nec_crit}.
These are the symmetries of local unitary invariants (including local spectra).
We will use both symmetries in the next section to formulate an invariant hierarchy of semidefinite programs.

\section{SDP refutation}\label{sec:SDPref}

\subsection{Primal and dual programs}
We follow Watrous~\cite{https://doi.org/10.48550/arxiv.1207.5726} and Doherty, Parrilo, and Spedalieri~\cite{PhysRevA.69.022308} to recall: 
a semidefinite program (SDP) is specified by a 
hermiticity preserving linear map $\Xi: L(\XX) \to L(\YY)$ and Hermitian operators $C$ and $D$.
Define the inner product $\langle A,B \rangle = \tr(A^\dag B)$, 
and denote the set of positive and hermitian operators on a Hilbert space $\HH$ by $\Herm(\HH)$ and $\Pos(\HH)$, respectively.
Then the primal and dual of the semidefinite program read

\noindent\begin{minipage}{.5\linewidth}
\begin{align}\label{eq:primal}
 &\underline{\mathrm{Primal:}}    \nn\\
 &\underset{X}{\mathrm{maximize}} && \langle C, X \rangle \nn\\
 &\mathrm{such~that}              && \Xi(X) = D    \\
 &&& X \in \Pos(\XX) \nn
\end{align}
\end{minipage}
\noindent\begin{minipage}{.5\linewidth}
\begin{align}\label{eq:dual}
 &\underline{\mathrm{Dual:}}    \nn\\
 &\underset{Y}{\mathrm{minimize}} && \langle D, Y \rangle  \nn\\
 &\mathrm{such~that}              && \Xi^*(Y) \geq C \\
 &&& Y \in \Herm(\YY) \nn
\end{align}
\end{minipage}

\smallskip
Operators $X$ and $Y$ that meet the constraints of ~\eqref{eq:primal} and \eqref{eq:dual}
are said to be primal and dual feasible, respectively. Denote the set of primal and dual feasible operators by
$\PP$ and $\DD$. 
Every semidefinite program satisfies {\em weak duality}, 
that is, for all $X\in\PP$ and $Y\in\DD$,
\begin{align}\label{eq:weak-duality}
 \langle D, Y\rangle - \langle C, X \rangle
 = 
 \langle \Xi(X), Y\rangle - \langle C, X \rangle
 = 
 \langle \Xi^*(Y) - C, X\rangle \geq 0\,.
\end{align} 

Interestingly, weak duality ~\eqref{eq:weak-duality} can be used to give an {\em SDP refutation} for the feasibility ($C=0$) 
of a primal problem: 
if there exists a feasible $Y \in \DD$
with $\langle D, Y\rangle < 0$, 
weak duality~\eqref{eq:weak-duality} is violated. 
This implies that the primal problem is infeasible.
The operator~$Y$ then provides a certificate of infeasibility.

\subsection{Primal hierarchy}
Incorporating the symmetries ~\eqref{eq:U-inv} and ~\eqref{eq:S-inv},
Proposition~\ref{prop:qmp2} can be formulated as a hierarchy of semidefinite programs 
for feasibility ($C=0$), indexed by~$k \in \N$.
\begin{equation}\label{eq:hierarchy_primal}
\begin{aligned}
 &\underline{\mathrm{Primal:}}    \\
 &\underset{X}{\mathrm{maximize}} && \langle 0, X\rangle\\
 &\mathrm{such~that} 
   && \tr(X) = 1 \\
 & && \tr\big( \eta(\sigma^A) X \big)
    = q_{A, \sigma}                            &&\forall A \in \AA \,,\quad \sigma \in S_k\\
 & && X = \tau(\pi)  X \tau(\pi)^{-1}       && \forall \pi \in S_k\\
 & && X = \frU^{-1} X \frU                  && \forall \frU = (U_1 \ot \dots \ot U_n)^{\ot k} \,: U_1, \dots, U_n \in \UU(d)\\
 & && X \in \Pos(\HH^{\otimes k})
 \end{aligned}
\end{equation}

\begin{remark}
 Technically speaking, the optimization program \eqref{eq:hierarchy_primal} and also its dual \eqref{eq:hierarchy_dual} below
 are not semidefinite programs due to the appearance of
 infinitely many constraints of the form $X = \frU^{-1} X \frU$.
 However, these conditions determine the commutant of a set of operators which is a linear subspace.
 Thus, the conditions translate into finitely many constraints (see also Section~\ref{sec:symmetry}).
\end{remark}

Note that some elements in $S_k$, for example $(12)(34)$,
are of the form $\sigma \times \sigma^{-1}$,
where $\times$ denotes the direct group product.
For these, the corresponding
constraints are ``quadratic'':
$\tr\big( \eta(\sigma^A) \otimes \eta(\sigma^A)^\dagger X \big) = q_{A, \sigma\times \sigma^{-1}}  = q_{A, \sigma}^2$.
This will be relevant for completeness of the hierarchy, which we show in Theorem~\ref{thm:complete}.
For now, we return to the question of feasibility of this program.

\smallskip
In ~\eqref{eq:hierarchy_primal} and using the notation of ~\eqref{eq:primal},
we write 
\begin{equation}\label{eq:ximap}
 \Xi(X)   = 
 \bigoplus_{\Aell} \Xi^{A,\sigma}(X)
\end{equation}
where the hermitian maps $\Xi^{A, \sigma}$ and their duals are given by 
\begin{align}
\Xi^{A, \sigma}(X)                &= \frac{1}{2} \tr\big( (\eta(\sigma^A) + \eta(\sigma^A)^\dag) X\big) \,,\nn\\
{(\Xi^{A, \sigma})}^*(y^{A,\sigma}) &= \frac{1}{2} y^{A, \sigma} \big(\eta(\sigma^A) + \eta(\sigma^A)^\dag \big)\,, \label{eq:xi_primal}
\end{align}
with associated constants
$D^{A,\sigma} = q_{A, \sigma}$.

\subsection{Dual hierarchy}
 
Consider now the dual of the hierarchy in Eq.~\eqref{eq:hierarchy_primal}.
We first start by identifying the symmetries present in the dual.
The objective function of the dual program is
\begin{equation}
  \langle D, Y\rangle = \langle \Xi(X), Y\rangle = \langle X, \Xi^*(Y)\rangle\,.
\end{equation} 
We can now apply the symmetries of $X$ to see that 
\begin{align}
\begin{aligned}
\langle X, \Xi^*(Y)\rangle &= \langle  \tau(\pi)  X \tau(\pi)^{-1},  \Xi^*(Y)\rangle = \langle X , \tau(\pi)^{-1}  \Xi^*(Y) \tau(\pi) \rangle \,,\\
\langle X, \Xi^*(Y)\rangle &= \langle \frU X \frU^{-1}, \Xi^*(Y)\rangle = \langle X , \frU^{-1} \Xi^*(Y) \frU \rangle \,,
\end{aligned}
\end{align} 
holds for all $\pi \in S_k$ and unitaries of the form $\mathfrak{U} = (U_1 \ot \dots \ot U_{n})^{\ot k}$.

Thus, we can write dual of the hierarchy in Eq.~\eqref{eq:hierarchy_primal}, indexed by $k\in \N$, as
\begin{equation}\label{eq:hierarchy_dual}
  \begin{aligned}
 &\underline{\mathrm{Dual:}}    \\
 &\underset{y^{A, \sigma}}{\mathrm{minimize}}
 &&
 \!\!\!
 \sum_{\Aell}     y^{A, \sigma} q_{A, \sigma} \\
 &\mathrm{such~that}
              &&  \Xi^* = \tau(\pi)  \Xi^* \tau(\pi)^{-1}    &&\forall \pi \in S_k\\
              &&& \Xi^* = \frU \Xi^* \frU^{-1}               &&\forall \frU = (U_1 \ot \dots \ot U_n)^{\ot k} : U_1, \dots, U_n \in \UU(d) \\
              &&& \Xi^* \in \Pos(\HH^{\otimes k})
\end{aligned}
\end{equation}
where 
\begin{equation*}
 \Xi^* = \Xi^*(Y) = \sum_{\Aell}  y^{A, \sigma} \eta(\sigma^A)
\end{equation*} 

\begin{remark}\label{rem:noncrossperm}
We say that an element $\sigma = (\sigma_1, \dots, \sigma_n) \in \C S_k^n$ {\em factorizes}
if 
the operator 
$\eta(\sigma)$ factorizes along the copies where its cycles act.
Thus, factorizing permutations can be evaluated by polynomials in~$q_{A, \ell}$.
For example, $\sigma = ((12), (12)(34), (34)) \in S_4^3$ yields
\begin{equation}
 \tr( \eta(\sigma) \varrho^{\ot 4}) = \tr(\varrho_{AB}^2) \tr(\varrho_{BC}^2) = q_{AB,2} \cdot q_{BC,2}\,.
\end{equation}
The dual program~\eqref{eq:hierarchy_dual} can be strengthened by replacing the sum over $\ell$-cycles 
by a sum over factorizing permutations. 
This becomes only relevant when $k \geq 4$,
as one readily sees that all factorizing permutations are $\ell$-cycles for $k\leq 3$.
\end{remark}

\subsection{SDP refutation}
If for some $k \in \N$ the dual program~\eqref{eq:hierarchy_dual} is feasible with $\langle D, Y\rangle < 0$, 
then by violation of weak duality in Eq.~\eqref{eq:weak-duality}, the primal problem must be infeasible. Consequently, 
by Proposition~\ref{prop:qmp2}, the spectra corresponding to $q_{A, \ell}$ are incompatible.
For moderate sizes, such semidefinite programs can be solved by a computer.\footnote{In order for the dual program to be numerically bounded, 
one can change the dual to a feasibility problem with the constraint $\langle D, Y\rangle =-1$.}

The SDP refutation for detecting incompatibility of prescribed spectra can now be understood in simple terms: 
suppose there exists a density matrix $\varrho$ with 
$\tr(\eta((\a_1 \dots \a_\ell)^A) \varrho^{\ot k})$ $= q_{A, \ell}$.
If one finds a positive semidefinite operator $F = \Xi^*(Y)$ 
satisfying the conditions in ~\eqref{eq:hierarchy_dual}
and for which
\begin{equation}\label{eq:incompatibility_witness}
 \tr(F\varrho^{\ot k} ) 
 = 
 \sum_{\Aell}  y^{A, \sigma} \tr\big( \eta(\sigma^A) \varrho^{\ot k} \big)
 = 
 \sum_{\Aell}  y^{A, \sigma} \, q_{A, \sigma} < 0
\end{equation} 
holds, then one has arrived at a contradiction, because 
the trace inner product of two semidefinite operators must be non-negative.

\section{Symmetry-reduction}\label{sec:symmetry}

Consider the symmetries appearing in Eq.~\eqref{eq:hierarchy_dual},
\begin{align}\label{eq:symmetries2}
\tau(\pi) \Xi^*(Y) \tau(\pi)^{-1}   &= \Xi^*(Y) \quad \forall \pi \in S_k \nn\\
\frU\, \Xi^*(Y) \frU^{-1}              &= \Xi^*(Y) \quad
\forall \frU = (U_1 \ot \dots \ot U_n)^{\ot k} : U_1, \dots, U_n \in \UU(d)\,.
\end{align} 
From the Schur-Weyl duality it follows that the actions commute, $[\tau(\pi), \frU] = 0$.

Let us now decompose $\cdkn$ under these symmetries.
Consider the collection of the $k$ first subsystems. 
By the Schur-Weyl duality, the space $\cdk$ decomposes as
\begin{equation}\label{eq:SW}
 \cdk \simeq \bigoplus_{\substack{\lambda \vdash k \\ \height(\lambda) \leq d}} \UU_\lambda \ot \SS_\lambda\,,
\end{equation}
where the unitary group acts on $\UU_\lambda$ and the symmetric group on $\SS_\lambda$.
Consequently, 
\begin{equation}\label{eq:SWn}
(\cdk)^{\otimes n} \simeq 
\bigotimes_{i=1}^n 
\Big(
\bigoplus_{\substack{\lambda_i \vdash k \\ \height(\lambda_i) \leq d}} 
\UU_{\lambda_i} \ot \SS_{\lambda_i}
\Big)  \,.
\end{equation}
An operator $X$ on $(\cdk)^{\otimes n}$ that is invariant under the symmetries~\eqref{eq:symmetries2} 
will have the form
\begin{equation}
X = \bigotimes_{i=1}^n 
\Big(
\bigoplus_{\substack{\lambda_i \vdash k \\ \height(\lambda_i) \leq d}} 
\one_{\lambda_i} \ot X_{\lambda_i}
\Big)  \,.
\end{equation}
Then $X\geq 0$ if and only if $X_{\lambda_1} \ot \dots \ot X_{\lambda_n} \geq 0$ for all $\lambda_1, \dots, \lambda_n \vdash k$ with $\height(\lambda_i) \leq d$.

Denote by $R_\lambda(\sigma)$ an irreducible orthogonal representation of $\sigma$ corresponding to the partition $\lambda \vdash k$.
For $ \sigma = (\sigma_1, \dots, \sigma_n) \in S_k^n$,
denote similarly 
\begin{equation}
 R_{\lambda_1, \dots, \lambda_n}(\sigma) := R_{\lambda_1}(\sigma_1) \ot \dots \ot R_{\lambda_n}(\sigma_n)\,.
\end{equation} 
 
We then have the following symmetry-reduction.
\begin{proposition}\label{prop:sym_red_dual}
In the dual program~\eqref{eq:hierarchy_dual}, it holds that
\begin{equation}
 \Xi^*(Y) = 
 \sum_{\Aell}  y^{A, \sigma} \eta^d(\sigma^A)
 \geq 0
\end{equation} 
if and only if
\begin{equation}
 F_{\lambda_1, \dots, \lambda_n} =  
 \sum_{\Aell}  y^{A, \sigma} R_{\lambda_1, \dots, \lambda_n}(\sigma^A )
 \geq 0
\end{equation}
for all $\lambda_1,\dots, \lambda_n \vdash k$ with $\height(\lambda_i) \leq d$.
\end{proposition}
\begin{proof}
 The variable $\Xi^*(Y)$ is positive semidefinite if and only if it is positive semidefinite in each of its isotypic components.
 By \eqref{eq:SWn}, the isotypic components of $\eta$ are labeled by the partitions
 $\lambda_1,\dots, \lambda_n \vdash k$ with $\height(\lambda_i) \leq d$.
 Both $\eta$ and $R_{\lambda_1, \dots, \lambda_n}$ are orthogonal representations and thus $\ast$-algebras.
 Thus, the map
 \begin{equation}
  \phi: \eta \to \bigoplus_{\substack{\lambda_i \vdash k \\ \height(\lambda_i) \leq d}} R_{\lambda_1, \dots, \lambda_n}
 \end{equation}
 is a $\ast$-isomorphism.
 As $\ast$-homomorphisms between $\ast$-algebras preserve positive semidefiniteness~\cite{Bachoc2012},
 this proves the claim.
\end{proof}

Proposition~\ref{prop:sym_red_dual} allows to find SDP refutations with a fewer number of variables but of equal strength than the 
naive approach of Eq.~\eqref{eq:hierarchy_dual}. 
This symmetry-reduced hierarchy can be stated as
\begin{equation}\label{eq:hierarchy_dual_symred}\tag{SDP-SC}
  \begin{aligned}
 &\underset{\{y^{A, \sigma}\}}{\mathrm{minimize}} &&
 \!\!\!\sum_{\Aell}  \!\!\!  y^{A, \sigma} q_{A, \sigma}\\
 &\mathrm{such~that}
 && \sum_{\Aell}  y^{A, \sigma} R_{\lambda_1, \dots, \lambda_n}(\sigma^A ) \geq 0 
 && \forall \lambda_1,\dots, \lambda_n \vdash k\,: \height(\lambda_i) \leq d,
\end{aligned}
\end{equation}
and we obtain:
\begin{theorem}\label{thm:hierarchy_dual_symred}
Let $\AA$ be a collection of subsets of $\{1,\dots,n\}$
with prescribed spectra of reductions $\mu = \{ \mu_A \, |\, A \in \AA \}$.
If a level in the hierarchy~\eqref{eq:hierarchy_dual_symred}
returns a negative value, then the spectra are incompatible with a joint quantum state on $(\C^d)^{\ot n}$.
\end{theorem}
\begin{proof}
 Proposition~\ref{prop:qmp2} states a necessary condition for spectral compatibility.
 Proposition~\ref{prop:sym_red_dual} allows for a symmetry-reduction of the corresponding SDP formulation ~\eqref{eq:hierarchy_primal}.
 A negative value in ~\eqref{eq:hierarchy_dual_symred} violates weak duality~\eqref{eq:weak-duality}. 
 Consequently, the putative marginal spectra 
 $\{ \mu_A \, |\, A \in \AA \}$ are then incompatible on $(\C^d)^{\ot n}$.
\end{proof}
Note that the symmetry-reduced hierarchy (\ref{eq:hierarchy_dual_symred}) is equivalent to the program (\ref{eq:hierarchy_dual})
while having a smaller number of variables.

\subsection{Scaling}

Given a collection of partitions $(\lambda_1, \dots, \lambda_n)$, the associated irreducible representation has dimension
$\prod_{i=1}^n \chi_{\lambda_i}(\id)$.
A Hermitian matrix of size $N\times N$ has $N^2$ real variables.
Accordingly, the symmetry-reduced SDP contains
\begin{equation}
 \frac{1}{2} \sum_{\substack{\lambda_1, \dots, \lambda_n \vdash k \\ \height(\lambda_i) \leq d}}
 \prod_{i=1}^n \chi_{\lambda_i}^2(\id)
\end{equation}
real variables.
Table~\ref{tab:dims} shows the relative growth of the naive unsymmetrized SDP versus that of the symmetrized SDP.

\ytableausetup{smalltableaux}
\begin{example}
 Consider three copies of a three-qubit state with associated space $((\C^2)^{\ot 3})^{\ot 3}$. 
 Under the action of $\UU(2)$, 
 the space $(\C^2)^{\ot 3}$ decomposes into irreducible representations (irreps)
 associated to the partitions $3=3$ and $2+1=3$,
 whose dimensions are $1$ and $2$, respectively.
 Thus, the full space carries the irreps
 \begin{align}
 \mathrm{irreducible} & \mathrm{~representation} & \mathrm{dimension} \nn\\
  \ydiagram{3}  &\ot \ydiagram{3}   \ot \ydiagram{3}    & 1\cdot 1 \cdot 1 = 1 \nn\\
  \ydiagram{3}  &\ot \ydiagram{3}   \ot \ydiagram{2,1}  & 1\cdot 1 \cdot 2 = 2 \nn\\
  \ydiagram{3}  &\ot \ydiagram{2,1} \ot \ydiagram{2,1}  & 1\cdot 2 \cdot 2 = 4 \nn\\
  \ydiagram{2,1} &\ot \ydiagram{2,1} \ot \ydiagram{2,1} & 2\cdot 2 \cdot 2 = 8
 \end{align}
 as well as permutations thereof. 
The total number of real variables in the symmetry-reduced space above is $125$,
fewer than the $(2^9)^2 = 262144$ real variables required for an SDP of nine qubits.
\end{example}

\section{Dimension-free incompatibility}\label{sec:dimfree}

We now show that when $k\leq d$, 
the incompatibility witnesses found by the hierarchy~\eqref{eq:hierarchy_dual_symred} are {\em dimension-free}.
That is, they certify incompatibility of spectra of joint states in arbitrary local dimensions.

\begin{theorem}\label{thm:dimfree}
When the number of copies is less or equal to the local dimension ($k\leq d$),
the incompatibility witnesses produced by the hierarchy~\eqref{eq:hierarchy_dual_symred} 
are dimension-free
and the detected spectra are incompatible in all local dimensions.
\end{theorem}

\begin{proof}
\noindent Let an incompatibility witness for dimension $d$ using $k$ copies be given.

First, consider the case of local dimension $d'>d$: 
Recall that the program~\eqref{eq:hierarchy_dual_symred} is equivalent to the program~\eqref{eq:hierarchy_dual}.
Now, consider a given incompatibility witness, i.e., a feasible
$ F = \Xi^*(Y)$
satisfying \eqref{eq:hierarchy_dual} with negative objective function.
It can be written as $F = \eta^d(f)$ with $f \in \C S_k^n$.
Because of $k \leq d$ and the Schur-Weyl decomposition~\eqref{eq:SWn}, 
$f \in \ker(\eta^d)^\perp$.
This implies two things:
First, 
because $F\geq 0$, there is an element $a \in \C S_k^n$ such that $f = aa^*$.
Consequently, if $F \geq 0$ then also $F' = \eta^{d'}(f) =\eta^{d'}(aa^*) \geq 0$ for all $d'$.
Second, the decomposition of $F$ and $F'$ into permutations is identical.
Thus, the expectation values 
$\tr\big(F \varrho^{\ot k} \big)$ 
and 
$\tr\big(F' {\varrho'}^{\ot k} \big)$ 
coincide for $\varrho$  and $\varrho'$ with spectra $\mu$.
Thus, if $F$ is an infeasibility certificate for spectra $\mu$ in dimension $d$
then $F'$ is an infeasibility certificate for spectra $\mu$ in dimension $d'$.

Now, we consider the case $d'<d$: 
Through the direct sum $\C^d = \C^{d'} \oplus \C^{(d-d')}$, the space $(\C^{d'})^{\ot n}$ embeds into $(\C^{d})^{\ot n}$.
Clearly, spectral compatibility in the smaller space $(\C^{d'})^{\ot n}$ implies compatibility in the larger space $(\C^{d})^{\ot n}$.
Consequently, incompatibility in $(\C^{d})^{\ot n}$ implies incompatibility in $(\C^{d'})^{\ot n}$.

Thus, if $k\leq d$ and
$F = \eta^{d}(f)$ certifies for spectra to be incompatible with a joint state on $(\C^{d})^{\ot n}$,
then the same spectra are also incompatible on $(\C^{d'})^{\ot n}$ with $d' \in \N_+$.
\end{proof}

For numerical calculations, this dimension-free property can be helpful:
any incompatibility witness found, as long as $k\leq d$, will certify the spectra
to be incompatible with a joint state with any local Hilbert space dimensions.

\section{Completeness and convergence}\label{sec:comp}

\subsection{Completeness}

We now show that the hierarchy~(\ref{eq:hierarchy_primal}) 
is complete, that is, feasible 
at every level of the hierarchy if and only if the spectra are compatible. 
For this, we use a strategy similar to that in a recent work by Ligthart and Gross~\cite{10.1063/5.0143792}, where 
de Finetti together with ``quadratic constraints'' yields completeness.\footnote{
We thank Laurens T. Ligthart for explaining to us their proof.}

The quantum de Finetti theorem states \cite{Christandl2007}: 
Suppose $\varrho_t \in L((\C^D)^{\ot t})$ is permutation-invariant and infinitely symmetrically extendable, 
that is, 
there exists $\varrho_k \in L((\C^D)^{\ot k})$ for every $k > t$, such that 
\begin{equation*}
 \quad\quad \tr_{k-t} (\varrho_k) = \varrho_t\,,
 \quad\quad\quad
 \tau(\pi) \varrho_k \tau(\pi)^{-1} = \varrho_k\,, \forall \pi \in S_k\,.
\end{equation*} 
Then 
\begin{equation*}
 \varrho_t = \int \varrho^{\ot t} \text{d}m(\varrho)\,.
\end{equation*}
for a measure $m$ on the set of states in $L(\C^D)$.

\begin{theorem}\label{thm:complete}
 The marginal spectra $\{ \mu_A \, |\, A \in \AA \}$ are compatible with a joint quantum state on $(\C^d)^{\ot n}$,
 if and only if 
 every level of the hierarchy~\eqref{eq:hierarchy_primal} is feasible.
\end{theorem}
\begin{proof}
``$\Longrightarrow$'': if the marginal spectra are compatible with a joint quantum state,
then $X_k = \varrho^{\otimes k}$ is feasible due to Proposition~\ref{prop:qmp2}.\\
``$\Longleftarrow$'':
Suppose the hierarchy~\eqref{eq:hierarchy_primal} is
is feasible at every level $k\in \N$. In particular, if level $k$ is feasible, then also level $k-1$.
Thus, there exists a sequence of feasible $\{\tilde X_k\}_{k=1}^\infty$,
such that $\tilde X_{k-1} = \tr_k(\tilde X_k)$
holds at every level of the hierarchy and $\tau(\pi) X_k \tau(\pi)^{-1} = X_k$
for all $k \in S_k$.
This means that every $\tilde X_k$ is infinitely symmetrically extendable.

Now consider the permutation $\sigma=(1...\ell)$ for $\ell \leq \lfloor \tfrac{k}{2}\rfloor$.
Then $\tilde X_k$ fulfils
the constraints that appear in the primal hierarchy \eqref{eq:hierarchy_primal} of the form
\begin{align}
\tr(\eta(\sigma^A) \tilde X_k) = q_{A, \sigma} = q_{A, \ell}\,, \nn\\
\tr(\eta(\sigma^A)\otimes {\eta(\sigma^A)}^\dag \tilde X_k) = q_{A, \sigma}^2 = q_{A,\ell}^2\,,
\end{align}
where we understand $\eta(\sigma^A)^\dag$ appearing above to act on a disjoint set of $\ell$ tensor factors
(e.g., on tensor factors $\ell +1$ to $2\ell$).
As a consequence,
\begin{align}\label{eq:quad1}
 \tr \big(
(\eta(\sigma^A) - q_{A, \sigma}\one)
\ot
( \eta(\sigma^A) - q_{A, \sigma}\one)
\tilde X_k \big)^\dag = 0\,.
\end{align}
As $\tilde X_k$ is infinitely symmetrically extendable, by
the quantum de Finetti theorem, the reduction of $\tilde X_k$ onto size $2\ell$ is separable as
\begin{equation}
 \tr_{k-2\ell} (\tilde X_k) = \int \varrho^{\ot 2\ell} \text{d}m(\varrho)\,.
\end{equation}
Then the constraint of Eq. \eqref{eq:quad1} factorizes as
\begin{align}\label{eq:square}
 &\int
 \tr \big(
(\eta(\sigma^A) - q_{A, \sigma}) \varrho^{\ot \ell}
\big)
\cdot
\tr \big(
(\eta(\sigma^A) - q_{A, \sigma})^\dag \varrho^{\ot \ell}
\big) \text{d}m(\varrho) \nn\\
&=
 \int
 {\big|
 \tr \big(
(\eta(\sigma^A) - q_{A, \sigma}) \varrho^{\ot \ell}
\big)
\big|^2} \text{d}m(\varrho) = 0
\end{align}
 This implies that $\tr(\eta(\sigma^A) \varrho^{\ot \ell}) = q_{A, \sigma}$ almost everywhere (w.r.t.~$m$), which means that there is a subset of quantum states in $L(\C^D)$ that is of full measure (w.r.t.~$m$) for which $\tr(\eta(\sigma^A) \varrho^{\ot \ell}) = q_{A, \sigma}$ is fulfilled exactly.
 The same reasoning holds for all $\sigma^A$ that constrain the spectrum.
 Consequently, the primal  hierarchy \eqref{eq:hierarchy_primal} is feasible for every $k \in \N$, if and only if a state compatible with the marginal spectra $\{ \mu_A \, |\, A \in \AA \}$ exists.
\end{proof}

\subsection{Convergence}

Suppose the primal SDP is feasible up to level $k$ in the hierarchy.
What guarantee   can be given for a state $\varrho$ to exist whose moments $\tr(\varrho^\ell_A)$
are close to the desired ones $q_{A,\ell}$?
A finite version of the quantum de Finetti theorem states that,
if the primal problem is feasible up to some level $k$ of the hierarchy,
then the state $\tr_{k-t}(\varrho_k)$ is close to a separable state~\cite{Christandl2007}:
Suppose $\varrho_t \in L((\C^d)^{\ot t})$ is permutation-invariant and symmetrically extendable for some $k > t$.
Then there exists a measure $m$ on the set of states in $L(\C^d)$, such that
\begin{equation}
 \Vert \varrho_t - \int \varrho^{\ot t} dm(\varrho) \Vert_1 \leq  \frac{2d^2t}{k}\,,
\end{equation}
where $\Vert X \Vert_1 = \frac12\tr|X| =  \frac12\tr \sqrt{X^\dag X}$ is the trace norm of $X$.
This allows us to show the following:
\begin{corollary}\label{cor:finitecomplete}
Let $\AA$ be a collection of subsets of $\{1,\dots,n\}$
with associated marginal spectra $\{ \mu_A \, |\, A \in \AA \}$.
If the level $k$ in the hierarchy~\eqref{eq:hierarchy_primal} is feasible, then there exists a state~$\varrho$ on $(\C^d)^{\ot n}$ such that for all $2\leq \ell \leq \lfloor \frac{k}2 \rfloor$,
\begin{align}
\vert \tr(\varrho_A^\ell) - q_{A,\ell}\vert^2 \leq \frac{12|\AA|d^{2n}}{k}(\ell - 1) (\ell + 2)\,.
\end{align}
\end{corollary}
\begin{proof}
We follow the strategy that if $X$
is close in trace distance to some
$Y = \int \varrho^{\ot k} dm(\varrho)$,
then the difference in their expectation values $|\langle \eta(\sigma^A) \rangle_X - \langle \eta(\sigma^A) \rangle_Y|$ for any $\ell$-cycle $\sigma^A$ is small.
By using quadratic constraints, this can be further strengthened,
such that $|\langle \eta(\sigma^A) \rangle_X - \langle \eta(\sigma^A) \rangle_{\varrho^{\otimes k}}|$
is small
for some $\varrho^{\otimes k}$
in the decomposition of
$\int \varrho^{\ot k} dm(\varrho)$. Finally, we consider the sum of squares over all $A$ and $j \leq l$ to show that there is a state close w.r.t.~all $A$ and powers $\ell$.

Let $\sigma = (\alpha_1\ldots\alpha_\ell)$ be some $\ell$-cycle. Then the primal feasible variable $X$ at level $2\ell \leq k$ of the hierarchy satisfies,
\begin{align}
\tr(\eta(\sigma^A) X) = q_{A, \ell}\,, \nn\\
\tr(\eta(\sigma^A)\otimes \eta(\sigma^A)^\dag X) = q_{A, \ell}^2\,,
\end{align}
Due to the finite quantum de Finetti theorem, there exists a measure $m$ such that
\begin{align} \label{eq:finiteqdf}
\Vert X - \int \varrho^{\ot 2\ell} dm(\varrho)\Vert_1 \leq  \frac{4d^{2n}\ell}{k} \,.
\end{align}
Let $Y = \int \varrho^{\ot 2\ell} dm(\varrho)$ and consider the expression
\begin{align}\label{eq:intsqrtoy}
  \int
 \big| \tr \big(
( \eta(\sigma^A) - q_{A, \ell}\one) \varrho^{\ot \ell} \big)\big|^2 dm(\varrho) 
&=\tr \big(
(\eta(\sigma^A) - q_{A, \ell}\one)
\ot
( \eta(\sigma^A) - q_{A, \ell}\one)^\dag
 Y \big)\,,
\end{align}
which is non-negative.

We now derive an upper bound for this expression. For this, observe that
\begin{align}\label{eq:long}
&\phantom{=|,}\tr \big(
(\eta(\sigma^A) - q_{A, \ell}\one)
\ot
( \eta(\sigma^A) - q_{A, \ell}\one)^\dag
 Y \big)\nn\\
&=
\phantom{|}\tr \big(
(\eta(\sigma^A) - q_{A, \ell}\one)
\ot
( \eta(\sigma^A) - q_{A, \ell}\one)^\dag
 (Y-X) \big) \nn\\
&= \phantom{|}\tr \big(\eta(\sigma^A) \ot \eta(\sigma^A)^\dag(Y-X) \big) - q_{A, \ell}\tr \big(\eta(\sigma^A) \ot \one (Y-X) \big) \nn\\
&\phantom{=}-q_{A, \ell}\tr \big(\one \ot \eta(\sigma^A)^\dag(Y-X) \big) + q_{A, \ell}^2 \tr \big(\one \ot \one (Y-X) \big) \nn\\
&=
\phantom{|}\tr \big(\eta(\sigma^A) \ot \eta(\sigma^A)^\dag(Y-X) \big) - q_{A, \ell}\tr \big(\eta(\sigma^A) \ot \one (Y-X) \big) \nn\\
&\phantom{=}-q_{A, \ell}\tr \big(\one \ot \eta(\sigma^A)^\dag(Y-X) \big)\,,
\end{align}
where we used the fact that
$\tr \big(\one \ot \one (Y-X) \big) = 0$ and that the whole expression vanishes on $X$.
Each term in \eqref{eq:long} is further bounded due to the Matrix Hölder inequality \cite{baumgartner2011inequality}:
for any unitary $U$, it holds that
\begin{align}
\vert \tr\big(X U\big) - \tr\big(Y U\big)\vert &\leq \sum_{i} s_i(X-Y) s_1\big(U\big) \nn\\
&=     \sum_i s_i(X-Y) \nonumber\\
&=     2\Vert X-Y\Vert_1\,,
\end{align}
where $s_i(U)$ denotes the $i$-th largest singular value of $U$ (which equals one for unitary matrices),
together with the identity $\Vert A \Vert_1 = \frac12 \sum_i s_i(A)$.
With this, \eqref{eq:long} is bounded by
\begin{align}\label{eq:upper}
 \tr \big(
(\eta(\sigma^A) - q_{A, \ell}\one)
\ot
( \eta(\sigma^A) - q_{A, \ell}\one)^\dag
 Y \big)
 &\leq (2 + 4q_{A,\ell}) ||Y - X||_1.
\end{align}
At this point, we use the finite quantum de Finetti theorem in Eq.~(\ref{eq:finiteqdf}). Together with the fact that $q_{A,\ell} \leq 1$, we get from Eqs.~(\ref{eq:intsqrtoy}) and (\ref{eq:upper}) that
\begin{align}\label{eq:upper2}
  \int \big| \tr \big( ( \eta(\sigma^A) - q_{A, \sigma}\one) \varrho^{\ot \ell} \big)\big|^2 dm(\varrho)
 &\leq \frac{24d^{2n}\ell}{k}
\end{align}
Note that the left-hand side of Eq.~(\ref{eq:upper2}) can be interpreted as the average of $\big| \tr \big( ( \eta(\sigma^A) - q_{A, \ell}\one) \varrho^{\ot \ell} \big)\big|^2$ over all $\varrho$ in the decomposition of $Y$. Thus, there must exist a $\varrho$ of non-zero measure such that
\begin{align}
\big|\tr(\varrho_A^\ell) - q_{A,\ell}\big|^2
&=
\big| \tr \big( ( \eta(\sigma^A) - q_{A, \sigma}\one) \varrho^{\ot \ell} \big)\big|^2
\leq \frac{24d^{2n}\ell}{k}.
\end{align}

A similar argument can be made for all spectra.
Consider the sum of the left-hand sides of Eqs.~(\ref{eq:long}) over all $A$ and $j\leq\ell$. Thus,
\begin{align}
  &\phantom{=}\int \sum_{A \in \AA}\sum_{j=2}^\ell \big| \tr \big( ( \eta( (1\ldots j)^A ) - q_{A, j}\one) \varrho^{\ot \ell} \big)\big|^2 dm(\varrho) \nn\\
  &= \sum_{A\in\AA} \sum_{j=2}^\ell \tr \big( (\eta((1\ldots j)^A) - q_{A, j}\one) \ot ( \eta( (j+1,\ldots,2j)^A ) - q_{A, j}\one)^\dag Y \big)\nn\\
  &\leq \frac{24|\AA|d^{2n}}{k} \sum_{j=2}^\ell j \nn \\
  &= \frac{12|\AA|d^{2n}}{k}(\ell - 1) (\ell + 2)\,.
\end{align}
Thus, there is again a state $\varrho$ in the decomposition of $Y$ with
\begin{align}
\sum_{A \in \AA}\sum_{j=2}^\ell \big| \tr \big( ( \eta( (1\ldots j)^A ) - q_{A, j}\one) \varrho^{\ot \ell} \big)\big|^2 \leq \frac{12|\AA|d^{2n}}{k}(\ell - 1) (\ell + 2)\,.
\end{align}
As the left-hand side of this inequality is a sum of positive terms, each of them must be bounded individually, yielding the claim.
\end{proof}

\section{Numerical results}\label{sec:numerics}

\subsection{Spectra of three-partite states}

\begin{figure}
 \centering
    \includegraphics[width=1\textwidth]{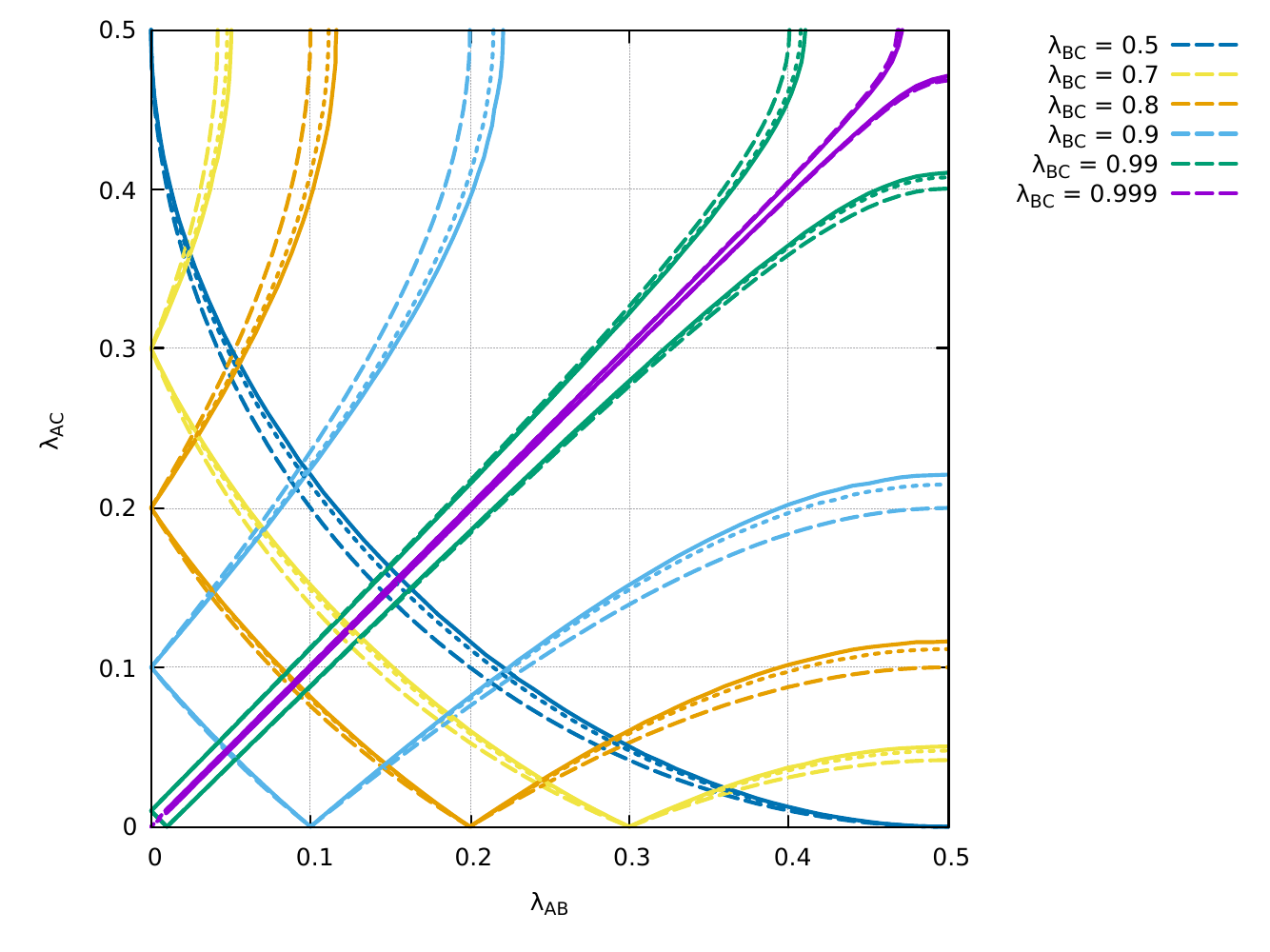}
    \caption{ \label{fig:threequbit_incompat_m4}
    {\bf Regions of spectral incompatibility.}
    Consider prescribed eigenvalues $\lambda_{AB}, \lambda_{AC}, \lambda_{BC}$ of rank-$2$ two-body marginals of three-partite states. 
    We plot the regions of certified incompatibility for values in the interval $[0, \frac12]$, 
    as the problem is symmetric under the exchange of $\lambda_{ij} \leftrightarrow 1-\lambda_{ij}$. 
    The infeasible regions are below (for $\lambda_{AC} < \lambda_{AB}$) and to the left (for $\lambda_{AC} > \lambda_{AB}$) of the lines.
    Shown are the boundaries of infeasibility for 
    $k=2$ (dashed lines),
    $k=4$ (dotted lines),
    and $k=4$ with factorizing permutations (solid lines), 
    with the height of the Young tableaux $d$ equal to the number of copies $k$. 
    Due Theorem~\ref{thm:dimfree}, the infeasibility regions are valid for tripartite states of arbitrary local dimensions.
    }
\end{figure}

As an example, consider a three-partite state $\varrho_{ABC}$ with two-body marginals $\varrho_{AB}, \varrho_{AC}$, and $\varrho_{BC}$
of rank two.
Their spectra are thus of the form
\begin{align}
 \spect(\varrho_{AB}) &=
 (\lambda_{AB}, 1-\lambda_{AB}) \nn\\
 \spect(\varrho_{AC}) &=
 (\lambda_{AC}, 1-\lambda_{AC}) \nn\\ 
 \spect(\varrho_{BC}) 
 &=
 (\lambda_{BC}, 1-\lambda_{BC})\,.
\end{align}
Evaluating the symmetry-reduced semidefinite programming hierarchy in Theorem~\ref{thm:hierarchy_dual_symred},
we obtain the incompatibility regions shown in 
Figure~\ref{fig:threequbit_incompat_m4}.
One sees that the use of four copies ($k=4$, dotted line) in the hierarchy
excludes a larger region of spectra than two only ($k=2$, dashed line). 
The use of factorizing permutations [see Remark \ref{rem:noncrossperm}] 
is even stronger ($k=4$, solid line). 

Recall that $d$ controls the height of Young tableaux used 
and that $k$ is the number of copies. 
In the symmetry-reduced formulation, the number of variables saturates when $d = k$ and contains fewer variables when $d < k$.
We choose the saturated parameters $k=d=2$ and $k=d=4$.
Due to Theorem~\ref{thm:dimfree}, our spectral incompatibility regions are valid for tripartite systems of arbitrary local dimensions.

A precise boundary of the region can be obtained through a divide and conquer algorithm with a precision of $10^{-3}$,
implemented with the Python interface PICOS~\cite{PICOS2022} and the solver MOSEK~\cite{mosek9.3}. 
The infeasibility boundaries are described by
\begin{align}\label{eq:two-copy}
&k=2: && \big(\lambda_{AB}-\frac12\big)^2 + \big(\lambda_{AC} - \frac12\big)^2 - \big(\lambda_{BC} - \frac12\big)^2 \leq  \frac14 \,, \nn\\
&k=4: && 7\Big(\big(\lambda_{AB}-\frac12\big)^2 + \big(\lambda_{AC}-\frac12\big)^2 - \big(\lambda_{BC}-\frac12\big)^2\Big) \nn\\
&& & - 2\Big(\big(\lambda_{AB}-\frac12\big)^4 + \big(\lambda_{AC}-\frac12\big)^4 - \big(\lambda_{BC}-\frac12\big)^4\Big) \leq \frac{13}{8} \,,
\end{align}
and two inequalities with exchanged roles of the parties.

\subsection{Purity inequalities}
The semidefinite programming hierarchy of Theorem \ref{thm:hierarchy_dual_symred} yields the infeasibility regions 
of Fig.~\ref{fig:threequbit_incompat_m4}.
Note that the SDP does not fix any individual eigenvalues but their power sums.
This can also be seen from the resulting incompatibility witnesses, 
which have the form $\sum_{\sigma \in S_k^n} y_\sigma \eta(\sigma)$. 
In fact, the optimal solution of the dual yields an optimal solution that does not depend on the precise choice of the eigenvalues, and $\tr(\varrho \otimes \varrho \sum_{\sigma \in S_k^n} y_\sigma \eta(\sigma)) \geq 0$ yields the following purity relations:
\begin{align}
 &k=2: && 1 - \tr(\varrho_{AB}^2) - \tr(\varrho_{AC}^2) + \tr(\varrho_{BC}^2) \geq 0 \,,\nn\\
 &k=4: && 1 - \frac1{20}\Big(15\tr(\varrho_{AB}^2) - 3\tr(\varrho_{AB}^4) + 15\tr(\varrho_{AC}^2) - 3\tr(\varrho_{AC}^4) \nn\\
      & && \phantom{ 1 - \frac1{20}\Big(}+9\tr(\varrho_{BC}^2) - 16\tr(\varrho_{BC}^3) + 3\tr(\varrho_{BC}^4)\Big) \geq 0\,.
\end{align}
These inequalities are valid for all tripartite states of arbitrary local dimension,
and correspond to the incompatibility witnesses

\begin{align}
&k=2: &&A \ot P \ot P   +   P \ot A \ot A, \nn\\
&k=4: &&\one - \frac1{20}\Big(15\eta\big((12)^{AB}\big) - 3\eta\big((1234)^{AB}\big) + 15\eta\big((12)^{AC}\big) - 3\eta\big((1234)^{AC}\big) \nn\\
&     &&\phantom{\one - \frac1{20}\Big(}+9\eta\big((12)^{BC}\big) - 16\eta\big((123)^{BC}\big) +3\eta\big((1234)^{BC}\big)\Big),
\end{align} 
where $P$ and $A$ are the projectors onto the symmetric and antisymmetric subspaces of $(\C^d)^{\ot 2}$.
We note that the inequality for $k=2$
is a linear combination of Rains' shadow inequalities~\cite{817508},
while the $k=4$ relation seems to be new.
Due to Theorem~\ref{thm:dimfree}, these inequalities hold for tripartite systems of arbitrary dimensions.

\subsection{Flat marginal spectra}
As a final example, 
we consider the two-body marginal spectra of pure three- and four-partite states. 
To simplify the discussion, we assume the two-body marginals to be flat, 
\begin{equation}
 \spect(\varrho_{S}) = \big(\frac{1}{r_S}, \dots, \frac{1}{r_S}, 0,\dots, 0\big)\,,
\end{equation} 
so that $r_S$ is the rank of the marginal on $S = \{i,j\}$ with $i\neq j$.
Some constraints on ranks are known: from the Schmidt decomposition,
it follows that tracing out a subsystem of dimension $d$ from a 
pure state yields a state of rank at most $d$. 
Additionally, Cadney et al.~\cite{CADNEY2014153} have conjectured the inequality
$r_{AB} r_{AC} \geq r_{BC}$.~\footnote{The conjecture is claimed to be proven in the preprint
~\cite{https://doi.org/10.48550/arxiv.2105.07679}.}

Let us apply the symmetry-reduced SDP hierarchy.
Consider the case of three-partite systems 
and fix the ranks $r_{AB}$, $r_{AC}$ and $r_{BC}$ and the local dimension $d$.
We ask whether the spectra are compatible with a pure joint state
and apply the SDP hierarchy of Theorem~\ref{thm:hierarchy_dual_symred} with $k=4$ and factorizing permutations. 
The nonexistence of pure states with flat marginal spectra is shown in Fig.~\ref{fig:flat_spectra} (top).
These numerical results agree with the known and conjectured rank inequalities.
In the case of four-partite states,
we fix ranks $r_{AB}$, $r_{AC}$ and $r_{AD}$ instead. 
Here, the hierarchy yields stronger results, shown in Fig.~\ref{fig:flat_spectra} (bottom).
In particular, depending on the local dimension,
we can exclude states with flat marginals and ranks $[r_{AB}, r_{AC}, r_{AD}]$ equal to 
$[3,2, 2]$, 
$[4,2,2]$,
$[4,3,2]$,
and
$[4,3,3]$.

It is interesting to see that there exist states that are excluded in dimension $d$, 
but that can be shown to exist in dimension $d'>d$.\footnote{
This is not in contradiction to Thm.~\ref{thm:dimfree}, 
as the incompatibility of these states was shown using $k>d$.}
This shows that our hierarchy can also obtain meaningful constraints on spectra
whose compatibility is dimension-dependent.

\section{Extensions}\label{sec:extensions}

We briefly sketch further extensions of our method.

\begin{enumerate}
\item {\bf Local unitary invariants and quantum codes.}
We sketch the construction of local unitary invariants as described by Rains~\cite{817508}.
Any polynomial in the coefficients of a matrix~$M$ can,
with suitable operators $F^{(k)}$, be written as
$
 \sum_k \tr(F^{(k)} M^{\otimes k})
$.
A unitary-invariant polynomial satisfies
\begin{align}
\sum_k \tr\big(F^{(k)} M^{\otimes k}\big)
&=
\sum_k \tr\big(F^{(k)} U^{\ot k}
M^{\otimes k}
(U^{\ot k})^{-1} \big)\,\quad \forall U \in \UU(d)\,.
\end{align}
This implies that
$F^{(k)} = (U^{\otimes k})^{-1} F^{(k)} U^{\otimes k}$ for all $U \in \UU(d)$.
By the Schur-Weyl duality [Eq.~\eqref{eq:SW}],
$F^{(k)}$ is spanned by elements that are
in a one-to-one correspondence with elements of $S_k$.
In other words, $F^{(k)} = \eta(\a)$ with $\a \in \C S_k$.
Considering a {\em local} unitary invariant polynomial,
the invariance is
\begin{equation}
 \tr(F M^{\otimes k}) =
 \tr(F \frU^{-1} M^{\otimes k} \frU )
\end{equation}
for all $\frU = (U_1 \ot \dots \ot U_n)^{\ot k} \,: U_1, \dots, U_n \in \UU(d)$.
With Eq.~\eqref{eq:SWn}, it follows that $F = \eta(\a)$ with $\a \in \C S_k^n$.

Consider now the problem of the compatibility of local unitary invariants with a quantum state.
Then the arguments of Section~\ref{sec:SDPref} and ~\ref{sec:comp},
but replacing $\sigma^A$ and $q_{A, \sigma}$ by a more general $\a \in \C S_k^n$ and $q_\sigma \in \C$,
establishes a complete semidefinite programming hierarchy
for the compatibility of a set of local unitary invariants.

This can be used to refute the existence of quantum codes with given parameters $(\!(n,K,\delta)\!)_d$:
a variant of the Knill-Laflamme condition states that a projector $\Pi$  on $(\C^d)^{\ot n}$
corresponds to a quantum code of distance~$\delta$,
if and only if
\begin{equation}
K B_j = A_j \quad j = 0,\dots, \delta-1\,.
\end{equation}
Here, the weight enumerators $A_j$ and $B_j$ are given by
\begin{align}
 A_j = \sum_{|E| = j} \tr(E \Pi) \tr(E^\dag \Pi)\,, \quad\quad
 B_j = \sum_{|E| = j} \tr(E \Pi E^\dag \Pi)\,.
\end{align}
Above, the sum spans over all elements $E$ of an orthonormal tensor-product basis of weight $|E|=j$
(e.g., the Pauli basis).
The $A_j$ and $B_j$ are both local unitary invariants~\cite{681316}.
Our hierarchy can then be used to rule out the existence of a compatible $\varrho = \Pi/K$,
certifying that a code with given parameters does not exist.

\smallskip
\item {\bf Equivariant state polynomials.}
Equivariant state polynomials form a type of polynomials whose variables are states,
and whose positivity is invariant under local unitary transformations.
To see how they are constructed from our formalism, note that
local unitary invariants of degree $k$ of $n$-partite quantum states
correspond to the expectation values of elements of $\C S_k^n$ on $\varrho^{\otimes k}$.
For example, for a bipartite state $\varrho$,
the element $(123)\times (132)\in S_3^2$ 
gives the invariant
\begin{equation}
\tr\big( (123)_1 \otimes (132)_2 (\varrho\ot \varrho\ot \varrho)\big)
= \tr\big(
(\varrho^{T_2} \varrho^{T_2})^{T_2} \varrho
\big)\,.
\end{equation}
More generally, the expectation values can also be taken with respect to
$\varrho_1^{k_1} \ot \dots \ot \varrho_m^{k_m}$. For our example above, one can form
expressions in two and three variables,
$\tr\big(\varrho^{T_2} \varrho^{T_2})^{T_2} \nu\big)$ and
$\tr\big(\varrho^{T_2} \sigma^{T_2})^{T_2} \nu\big)$.

Equivariant state polynomials are obtained by varying over a state that is linear in such expression.
More precisely,
every non-negative unitary invariant that is linear in at least one state
is in a one-to-one correspondence with a positive semidefinite equivariant state polynomial.
To see this, note that
\begin{align}\label{eq:trace_ineq}
 \tr\big(\eta_d(\a) \varrho_{1}^{\ot k_1} \ot \dots \ot \varrho_{m}^{\ot k_m} \ot \nu \big ) \geq 0
\end{align}
for all states $\varrho_1,\dots, \varrho_m,\nu$ if and only if the following is a positive semidefinite matrix,
\begin{align}\label{eq:psd_ineq}
 \tr_{1\dots m-1}\big(\eta_d(\a) \varrho_{1}^{\ot k_1} \ot \dots \ot \varrho_{m}^{\ot k_m} \ot \one \big ) \geq 0\,.
\end{align}
This follows from the self-duality of the positive cone, $A \geq 0$ if and only if $\tr(AB)\geq 0$ for all $B \geq 0$,
and the defining property of the partial trace, $\tr(M (N \ot \one)) = \tr( \tr_2(M) N)$.
Thus, to determine whether Eq.~\eqref{eq:psd_ineq} is positive semidefinite for all $\varrho_1, \dots, \varrho_m$,
it is enough to minimize Eq.~\eqref{eq:trace_ineq}
over all states $\varrho_1, \dots, \varrho_m, \nu$.
This can be done by taking $\tr(\eta(\sigma)X)$ as objective function in Eq.~\eqref{eq:hierarchy_primal}
with a relaxed symmetry constraint on $X$:  one does not aim to approximate $\varrho^{\ot k}$,
but $\varrho_1^{k_1} \ot \dots \ot \varrho_m^{k_m} \ot \nu$, and thus $X$ is invariant under the permutation
of the subsystems corresponding to $\varrho_1^{\ot k_1},\dots, \varrho_m^{\ot k_m}$ individually.
This yields a complete hierarchy to optimize over equivariant state polynomials under equivariant state polynomial or local unitary constraints.
With this one can systematically search for new purity and moment inequalities,
relevant, for example, for the task of entanglement detection~\cite{huber2020positive,PhysRevLett.132.070202}.
\end{enumerate}

\section{Related work}
In recent years, extensive work has approached the
quantum marginal problem~\cite{1742-6596-36-1-014, Schilling2014},
developing constraints on
operators~\cite{Butterley2006},
von Neumann entropies~\cite{0806.2962v1, doi:10.1063/1.4808218},
purities~\cite{PhysRevA.98.052317}
and
ranks~\cite{CADNEY2014153} of subsystems.
The perhaps most systematic approach to date uses representation theory of the symmetric group~\cite{quant-ph/0409113v1, Christandl2006}
and generalizes the polygon inequalities~\cite{quant-ph/0309186v2, PhysRevLett.90.107902}.
Here, we highlight
the existence of
critical~\cite{Bryan2018, Bryan2019locallymaximally}
and
absolutely maximally entangled states~\cite{PhysRevLett.118.200502, Yu2021, PhysRevLett.128.080507}
as guiding problems, achieving extreme values in spectra and entropies.
This has also led to the development of methods to reconstruct the joint state from
partial information~\cite{https://doi.org/10.48550/arxiv.2209.14154, 10.1088/1367-2630/abe15e},
to tackle the question of uniqueness in reconstruction~\cite{PhysRevA.99.062104, PhysRevA.96.010102, Klassen_phd-thesis},
to detect entanglement from partial information~\cite{PhysRevA.98.062102, Navascues2021entanglement},
and to investigate marginals of random states~\cite{Christandl_2014, Collins_2023}.
Fermionic settings are treated in
~\cite{PhysRevLett.108.263002, https://doi.org/10.48550/arxiv.2105.06459}.
For bi- and tripartite systems, complete lists of inequalities for non-overlapping spectra are
given in Refs.~\cite{1742-6596-36-1-014, VergneWalter2017}.
Ref.~\cite{munne2024sdpboundsquantumcodes} provides a complete hierarchy for the non-existence of quantum codes
using the state polynomial optimization framework.

The key systematic approach to overlapping marginals of spin systems
is that of symmetric extensions~\cite{PhysRevA.90.032318, Yu2021}.
Our work is inspired by Hall~\cite{PhysRevA.75.032102}, Yu et al.~\cite{Yu2021}, and Huber et al.~\cite{huber2022dimensionfree}.
However, these are neither applicable to the spectral formulation of the problem nor can they give results that are
dimension-free.

\section{Conclusions}
Our main result, Theorem~\ref{thm:hierarchy_dual_symred}, 
combines the techniques of
symmetric extension and symmetry reduction
to certify the incompatibility of marginal spectra.
This simple idea turns out to be quite powerful,
allowing for a complete hierarchy for spectral compatibility in arbitrary local dimensions
(Theorem~\ref{thm:dimfree}). 
At the same time, it can be used to differentiate different dimensions with respect to spectral compatibility:
There exist spectra which are non-trivially incompatible in dimension $d$, but compatible in $d'>d$.

We stress that our hierarchy is applicable not only to to reformulations of the spectral marginal problem
such as non-vanishing Kronecker coefficients and sums of hermitian matrices,
but also to the compatibility of local unitary invariants and the existence of quantum codes.
Finally, we believe that the equivariant state polynomial optimization
framework sketched in Section~\ref{sec:extensions} could find further applications.

A natural question is how to include in the hierarchy~(\ref{eq:hierarchy_dual})
constraints arising from a positive partial transpose, which could strengthen the symmetric extension hierarchy.
A symmetry-reduction similar to the one employed here 
would require a decomposition of the Brauer algebra.
For the case of $k=3$ copies, the Brauer Algebra
can be expressed as a linear combination of elements from
$\C S_3$ ~\cite{PhysRevA.63.042111} 
and it should be possible to formulate a semidefinite program analogous
to the level $3$ in Theorem~\ref{thm:hierarchy_dual_symred}.

\begin{acknowledgments}
FH was supported
by the Foundation for Polish Science through TEAM-NET (POIR.04.04.00-00-17C1/18-00),
by the National Science Centre Poland 2024/54/E/ST2/00451,
and
by the Polish National Agency for Academic Exchange under the Strategic Partnership Programme grant BNI/PST/2023/1/00013/U/00001.
For the purpose of Open Access, the author has applied a CC-BY public copyright licence to any Author Accepted Manuscript (AAM) version arising from this submission.

NW acknowledges support by the 
QuantERA grant QuICHE via
the German Ministry of Education and Research (BMBF grant no. 16KIS1119K).
We thank the
Dmitry Grinko,
Marcus Huber, 
Robert König, 
Felix Leditzky,
Simon Morelli, 
Saverio Pascazio,
Claudio Procesi,
Micha{\l} Studzinski,
Michael Walter,
and Andreas Winter for fruitful discussions
and the TQC referees for valuable comments.
We thank Laurens T. Ligthart for the idea of proving completeness using quadratic constraints.
Computational infrastructure and support were provided by the Centre for Information and Media Technology at Heinrich Heine University Düsseldorf.
\end{acknowledgments}

\newpage

\begin{table}
\begin{tabular}{@{}ccllll@{}}
\toprule
system & copies & $N_\mathrm{naive}$ & $N_\mathrm{sym}$ &\# blocks & max. size\\
\midrule
2 qubits
            & 2 & $\approx 2.6 \cdot 10^{2}$ & 4 & 4 & 1 \\
            & 3 & $\approx 4.1 \cdot 10^{3}$ & 25 & 4 & 4 \\
            & 4 & $\approx 6.6 \cdot 10^{4}$ & 196 & 9 & 9 \\
            & 5 & $\approx 1.0 \cdot 10^{6}$ & 1764 & 9 & 25 \\
3 qubits
            & 2 & $\approx 4.1 \cdot 10^{3}$ & 8 & 8 & 1 \\
            & 3 & $\approx 2.6 \cdot 10^{5}$ & 125 & 8 & 8 \\
            & 4 & $\approx 1.7 \cdot 10^{7}$ & 2744 & 27 & 27 \\
            & 5 & $\approx 1.1 \cdot 10^{9}$ & 74088 & 27 & 125 \\
4 qubits
            & 2 & $\approx 6.6 \cdot 10^{4}$ & 16 & 16 & 1 \\
            & 3 & $\approx 1.7 \cdot 10^{7}$ & 625 & 16 & 16 \\
            & 4 & $\approx 4.3 \cdot 10^{9}$ & 38416 & 81 & 81 \\
            & 5 & $\approx 1.1 \cdot 10^{12}$ & 3111696 & 81 & 625 \\
5 qubits
            & 2 & $\approx 1.0 \cdot 10^{6}$ & 32 & 32 & 1 \\
            & 3 & $\approx 1.1 \cdot 10^{9}$ & 3125 & 32 & 32 \\
            & 4 & $\approx 1.1 \cdot 10^{12}$ & 537824 & 243 & 243 \\
            & 5 & $\approx 1.1 \cdot 10^{15}$ & 130691232 & 243 & 3125 \\
\midrule
2 qutrits
            & 3 & $\approx 5.3 \cdot 10^{5}$ & 36 & 9 & 4 \\
            & 4 & $\approx 4.3 \cdot 10^{7}$ & 529 & 16 & 9 \\
            & 5 & $\approx 3.5 \cdot 10^{9}$ & 10609 & 25 & 36 \\
3 qutrits
            & 3 & $\approx 3.9 \cdot 10^{8}$ & 216 & 27 & 8 \\
            & 4 & $\approx 2.8 \cdot 10^{11}$ & 12167 & 64 & 27 \\
            & 5 & $\approx 2.1 \cdot 10^{14}$ & 1092727 & 125 & 216 \\
4 qutrits
            & 3 & $\approx 2.8 \cdot 10^{11}$ & 1296 & 81 & 16 \\
            & 4 & $\approx 1.9 \cdot 10^{15}$ & 279841 & 256 & 81 \\
            & 5 & $\approx 1.2 \cdot 10^{19}$ & 112550881 & 625 & 1296 \\
\midrule
2 ququarts
            & 4 & $\approx 4.3 \cdot 10^{9}$ & 576 & 25 & 9 \\
            & 5 & $\approx 1.1 \cdot 10^{12}$ & 14161 & 36 & 36 \\
3 ququarts
            & 4 & $\approx 2.8 \cdot 10^{14}$ & 13824 & 125 & 27 \\
            & 5 & $\approx 1.2 \cdot 10^{18}$ & 1685159 & 216 & 216 \\
4 ququarts
            & 4 & $\approx 1.8 \cdot 10^{19}$ & 331776 & 625 & 81 \\
            & 5 & $\approx 1.2 \cdot 10^{24}$ & 200533921 & 1296 & 1296 \\
\bottomrule
\end{tabular}
\vspace{1.4em}
\caption{
Number of real variables in the naive and symmetry-reduced SDP.
For the symmetry-reduced SDP the number of blocks and the size of the largest block is shown.
In comparison, an SDP size commonly solvable on modern laptops 
is that of a seven-qubit density matrix with $16384$ real variables.
Note that $d$ controls the height of Young tableaux used
and $k$ is the number of copies.
In the symmetry-reduced formulation, the number of variables saturates when 
$d=k$; containing fewer variables when $d<k$.
\label{tab:dims}}
\end{table}

\begin{table}[tpb]
 \centering
    \includegraphics[width=0.93\textwidth]{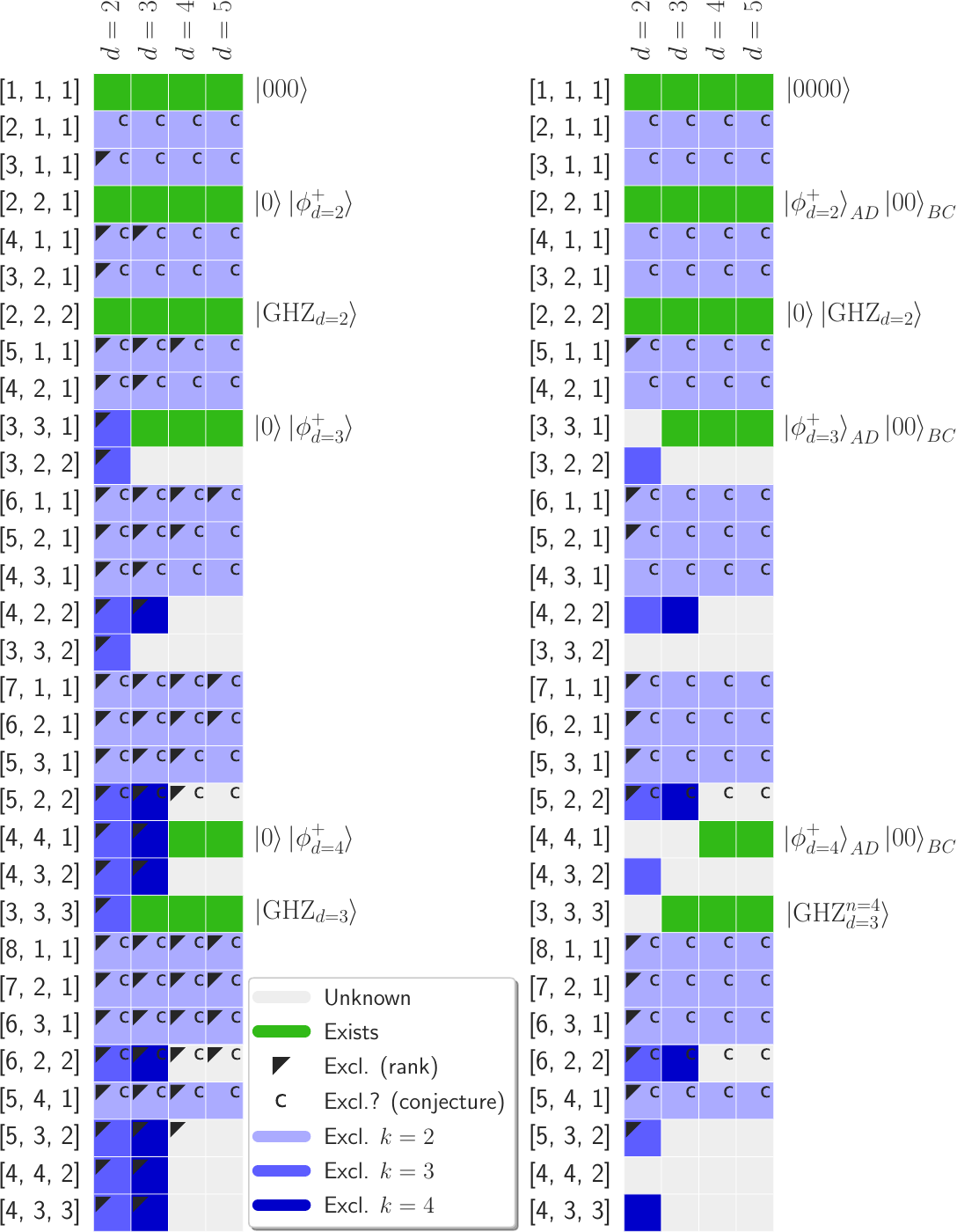}
    \caption{\label{fig:flat_spectra}
    {\bf Compatibility of two-body marginals with flat spectra.}
    We apply the hierarchy of Thm.~\ref{thm:hierarchy_dual_symred} for $k\in\{2,3,4\}$ 
    to two-body marginals with flat spectra. 
    Cases are excluded by incompatibility certificates (blue), 
    by known rank constraints (black triangle), and by conjectured rank constraints (c).
    Also shown are compatible pure states (green). 
    {Left:} Fix the ranks $[r_{AB},r_{AC},r_{BC}]$ and assume flat spectra for the two-body marginals of three-partite pure states of different local dimensions. Our hierarchy coincides with known and conjectured rank inequalities.
    {Right:} fix the ranks $[r_{AB},r_{AC},r_{AD}]$ and assume flat spectra of the two-body marginals of four-partite pure states.
    Additional cases $[3,2,2], [4,2,2], [4,3,2], [4,3,3]$ are excluded depending on the local dimension.
}
\end{table}

\bibliographystyle{quantum}
\bibliography{current_bib}
\end{document}